\documentclass[11pt,a4paper,english]{article}

\title{Isospectral Deformations in QFT: The Massive case} 
\author{A. Much\\ \footnotesize{Instituto de Ciencias Nucleares, UNAM, M\'exico D.F. 04510, M\'exico}}

\usepackage[pdftex]{graphicx} 
\usepackage{float}            
\usepackage[]{babel}
\usepackage[T1]{fontenc}   
\usepackage[cmex10]{amsmath}  
\usepackage{amstext}          
\usepackage{mathrsfs}
\usepackage{amsfonts}         
\usepackage{amssymb}          
\usepackage{setspace}
\usepackage{amsthm}
\usepackage{amsmath}
\usepackage{bbold}
\usepackage{bm}               
\usepackage{enumerate}        
\usepackage[bottom]{footmisc} 
\usepackage{array}            
\usepackage{textcomp}         
\usepackage{pdfpages}         
\usepackage{parskip}          
\usepackage[right]{eurosym}   
\usepackage{xcolor}           
\usepackage[square,numbers]{natbib}	  

\DeclareMathOperator{\arcsinh}{arcsinh}

\definecolor{darkred}{rgb}{0.7,0.0,0.0}

\newtheorem{theorem}{\textsc{Theorem}}[section]
\newtheorem{lemma}{\textsc{Lemma}}[section]
\newtheorem{proposition}{\textsc{Proposition}}[section]
 
\theoremstyle{definition} 
\newtheorem{definition}{\textsc{Definition}}[section]

\theoremstyle{remark}
\newtheorem{remark}{Remark}[section]
\newtheorem{example}{Example}[section]

\usepackage{fancyhdr}
 \numberwithin{equation}{section} 
\begin{document}
\maketitle
\abstract{We consider isospectral deformations of quantum field theories by using the
novel construction tool of warped convolutions. The deformation enables us to
obtain a variety of models that are wedge-local  and have nontrivial scattering matrices. } 
  \tableofcontents

\section{Introduction} 

 The equivalence of scalar fields living on a constant Moyal-Weyl space-time and wedge-local fields was proven in the ground breaking paper \cite{GL1}.  This important insight has since its publication generated a vivid interest in  algebraic and constructive QFT \cite{GL2, BS,  BLS,A,  GL4, Mor, MUc, BC, GL5}.\\\\In this context  an important tool has been formulated \cite{BS,BLS}, known as warped convolutions,  in order to deform quantum fields in a mathematical rigorous fashion. This mathematical precise formulation of deformation theory was used  to construct new QFT-models from a free theory at hand, \cite{A,  GL4, Mor, MUc, BC, GL5}. In particular the investigation shows that the newly obtained models have non-trivial scattering matrices, which even satisfy weakened  relativistic locality and covariance properties. The weakened properties are interesting from a physical point of view since relativistic symmetries are hard to realize when the notion of a point (non-commutative geometry) is lost.
 The technique of warped convolutions has been as well used in a quantum mechanical context in order to obtain   quantum mechanical effects and attack the quantum measurement problem, \cite{AA, Muc1}. \\\\
To obtain scalar fields living on a constant Moyal-Weyl space-time by using warped convolutions, one uses as generator of  deformation the  momentum operator. Hence, fields deformed with the \textbf{momentum operator} correspond to wedge local fields which are equivalent to fields living on a constant Moyal-Weyl space-time. 
Now the question which is the subject of investigation in this paper is the following: Are there more wedge-local fields corresponding to NC space-times that differ from the 
constant Moyal-Weyl space-time? An answer to this problem was given on the basis of a concrete example \cite{MUc}, where the special conformal operators were used instead of the momentum operators for deformation. The proofs in the case of the conformal transformations used the unitary equivalence to the momentum operator. Is this program extendable?  
Hence, can we take operators that are unitary equivalent to the momentum operator and obtain wedge-local fields on one hand and excitations living on a nontrivial non-commutative space-time on the other hand? A detailed answer will be given in this paper. We first define the operators on the appropriate domains and investigate how far the program of wedge-locality can be achieved.\\\\
The investigation shows that the Wightman properties of scalar fields deformed with the unitary equivalent operator are satisfied without any restriction. However, concerning the wedge-locality we obtain an additional requirement on the operators used for deformation.   
Moreover,  to  use the concept of tempered polarization free generators  \cite{Sch97,BBS} for scattering, we are obliged to show temperateness and polynomial boundedness of the field. The award of this concept is the ability   to calculate an explicit two-particle scattering of the deformed theory.  
Next to investigating locality, covariance and  scattering   we study the relation of deformed fields to a non-commutative space-time. The first step in this direction is to construct an isomorphism from the deformed $^*-$algebras to the unitary transformed $^*-$algebras of fields living on Moyal-Weyl. By using   simple examples it is further shown how the deformed  $^*-$algebras relate to the twist-deformation framework, i.e. to fields defined on a  non-commutative space-time. 
\\\\
The paper is organized as follows: in the second Section we lay out the novel tool of warped convolutions with all the definitions,  lemmas and  propositions needed for this work.  In Section $3$ we define the operators used for deformation and prove that the warped convolutions of the free scalar field, given as an oscillatory integral, is well-defined. Furthermore, the Wightman properties and wedge-locality are proven for a specific set of operators. The section ends with a whole class of examples. The fourth section describes the scattering of the constructed models. Last but not least we turn to the heart piece of this paper: The investigation of the resulting non-commutative space-times generated by isospectral deformations.

 \section{Warped convolutions in QFT}\label{s2}
In  this section we write all   basic definitions and lemmas of the deformation known under the name of warped convolutions. For  proofs of the respective lemmas we refer the reader to the original works \cite{BLS, GL1}. \newline\newline 
The authors start their investigation with a $C^{*}$-dynamical system $(\mathcal{A},\mathbb{R}^d)$, \cite{P97}.
It consists of a $C^{*}$-algebra $\mathcal{A}$  equipped with a strongly continuous automorphic action of the group $\mathbb{R}^d$ which will be denoted by $\alpha$.
Furthermore, let  $\mathcal{B}(\mathscr{H})$ be the Hilbert space of bounded
operators on $\mathscr{H}$ and let  the adjoint action $\alpha$ be implemented  by the weakly continuous unitary representation $U$. Then,  it is argued that  since the unitary representation $U$ can be extended to the algebra $\mathcal{B}(\mathscr{H})$,   there is no lost of generality   when one proceeds to  the $C^{*}$-dynamical system  $(C^{*}\subset  \mathcal{B}(\mathscr{H}),\mathbb{R}^d)$.  Here $C^{*}\subset  \mathcal{B}(\mathscr{H})$  is the $C^{*}$-algebra of all operators on which $\alpha$ acts strongly continuously.
\newline\newline 
Hence, we start by assuming the
existence of a strongly continuous unitary group $U$ that is a representation of the additive
group $\mathbb{R}^{d}$, $d\geq2$, on some separable Hilbert space $\mathscr{H}$. 
Moreover, to define the deformation of operators belonging to a $C^{*}$-algebra $C^{*}\subset  \mathcal{B}(\mathscr{H})$, we consider
elements belonging to the sub-algebra $C^{\infty}\subset C^{*} $. The sub-algebra $C^{\infty}$ is
defined to be  the $*-$algebra of smooth elements (in the norm topology) with respect to $\alpha$, which
is the adjoint action of a weakly continuous
unitary representation $U$ of $\mathbb{R}^{d}$ given by  
\begin{equation*}
 \alpha_{x}(A)=U(x)\,A\,U(x)^{-1}, \quad x \in \mathbb{R}^{d}.
\end{equation*} 
Let $\mathcal{D}$ be the dense domain of vectors in $\mathscr{H}$ which transform
smoothly under the adjoint action of $U$.
Then, the warped convolutions  for operators  $A\in C^{\infty}$ are given by
the following definition.\\
\begin{definition}
 Let $\theta$ be a real skew-symmetric matrix relative to the chosen bilinear form on $\mathbb{R}^d$, let $A\in C^{\infty}$ and let $E$ be
the spectral resolution of the unitary operator $U$.  Then, the corresponding warped convolution $A_{\theta}$
of $A$ is defined on the domain $\mathcal{D}$ according to
\begin{equation}\label{WC}
 A_{\theta}:=\int \alpha_{\theta x}(A)dE(x),
\end{equation}
where $\alpha$ denotes the adjoint action of $U$ given by $\alpha_{k}(A)=U(k)\,A\,U(k)^{-1}$.
\end{definition}
The restriction in the choice of operators is owed to the fact that the deformation is performed
with operator valued integrals. Furthermore,  one can represent the warped
convolution of $
 {A} \in  {C}^{\infty}$    by $\int \alpha_{\theta x}(A) dE(x)$ or $\int dE(x)
\alpha_{\theta x}(A) $, on the dense domain
$\mathcal{D}\subset\mathscr{H}$ of vectors smooth w.r.t. the action of $U$,  in terms of
strong limits 
\begin{equation*}
\int\alpha_{\theta x}(A) dE(x)\Phi=(2\pi)^{-d}
\lim_{\epsilon\rightarrow 0}
\iint  d x\, d y \,\chi(\epsilon x,\epsilon y )\,e^{-ixy}\,   \alpha_{\theta x}(A)\, U(y) \Phi,  
\end{equation*}
where $\chi \in\mathscr{S}(\mathbb{R}^d\times\mathbb{R}^d)$ with $\chi(0,0)=1$.
This representation makes calculations and proofs concerning the existence of  integrals
easier. In this work we  use both representations. \newline\newline 
The following lemma shows first that the two different warped
convolutions are equivalent. Second, it shows how the complex conjugation acts on
the warped convoluted operator.\newline
\begin{lemma}\label{wcl2}
 Let $\theta$ be a real skew symmetric matrix on $\mathbb{R}^d$ and let $  {A} \in
 {C}^{\infty}$.
Then 
\begin{enumerate} \renewcommand{\labelenumi}{(\roman{enumi})}
 \item   $\int \alpha_{\theta x}(A)dE(x)=\int dE(x)\alpha_{\theta x}(A)$
 \item $\left(\int \alpha_{\theta x}(A)dE(x)\right)^{*}\subset\int \alpha_{\theta
x}(A^{*})dE(x)$
\end{enumerate}
\end{lemma} 
Moreover, let us introduce the deformed product, also known as the Rieffel product
\cite{RI} 
by using warped convolutions. The two deformations are interrelated since   warped convolutions supply isometric
representations of Rieffel's strict deformations of $C^{*}$-dynamical systems with
actions of $\mathbb{R}^d$.
 \newline
\begin{lemma}\label{l2.1}
Let $\theta$ be a real skew-symmetric matrix on $\mathbb{R}^d$ and let $  {A},   {B} \in
 {C}^{\infty}$. Then 
\begin{equation*}
 {A}_{\theta}   {B}_{\theta} \Phi= (A\times_{\theta }B)_{\theta}\Phi, \qquad
\Phi\in\mathcal{D}.
\end{equation*}
where $\times_{\theta}$ is known as the Rieffel product
on ${C}^{\infty}$ and is given by, 
\begin{equation}\label{dp0}
(A\times_{\theta}B )\Phi=(2\pi)^{-d}
\lim_{\epsilon\rightarrow 0}
\iint  d x\, d y \,\chi(\epsilon x,\epsilon y )\,e^{-ixy} \, \alpha_{\theta x}(A)\alpha_{y}(B)\Phi.
\end{equation}
\end{lemma} 
The next proposition gives the transformation property of the warped convolution of an
operator under the adjoint action of a unitary or anti-unitary operator on $\mathscr{H}$.
This is relevant since in Section \ref{s32} we examine the transformation
properties of  deformed operators under Poincar\'{e} transformations.  \newline
\begin{proposition}\label{blsp1}
 Let $W$ be a unitary or anti-unitary operator on $\mathscr{H}$ such that
$WU(x)W^{-1}=U(Mx)$, $x\in
\mathbb{R}^d$, for some invertible matrix $M$. Then, for ${A} \in
 {C}^{\infty}$, 
\begin{equation*}
 WA_{\theta}W^{-1}=(WAW^{-1})_{\sigma M \theta M^{T}},
\end{equation*}
where $M^{T}$ is the transpose of M w.r.t the chosen bilinear form, $\sigma=1$ if $W$ is
unitary
and $\sigma=-1$ if $W$ is anti-unitary.
\end{proposition}
By using the former proposition and the homomorphism given in \cite{GL1} we
relate  skew symmetric matrices $\theta$ to wedges
$\mathcal{W}$. This in particular means that to
each deformed operator with deformation matrix $\theta$ there is a corresponding wedge
$\mathcal{W}$. 
\\\\
Most crucial to proving that the deformed fields satisfy a
weakened locality known as wedge locality is the following proposition. \newline
\begin{proposition}\label{pwlf}
 Let $A,B \in  {C}^{\infty}$ be operators such that $[\alpha_{\theta
x}(A),\alpha_{-\theta y}(B)]=0$ for all $x,y \in sp \,U$. Then 
\begin{equation*}
 [A_{\theta}, B_{-\theta}]=0.
\end{equation*}
\end{proposition}

In the next section  we adopt Formula (\ref{WC}) to define   warped convolutions with an unbounded
operator $X$ that is unitary equivalent to the momentum operator. Since we deform unbounded operators we are obliged to prove  that the deformation
formula, given as an oscillatory integral, is well-defined.  This is the subject of the next section.

\section{Deforming the scalar quantum field}
In this Section, we investigate the effect of deformation directly on a  free scalar field. The
unitary group  used for deformation, is given by the infinitesimal generator  $X$ that is unitary equivalent to the momentum operator.  Due to the unitary equivalence, the vector
operator $X$ is an essential self-adjoint operator on a dense domain and
therefore defines a strongly continuous unitary group that we denote by
$V(b):=e^{i b X }$. Furthermore, by using this Abelian group, an adjoint action can be
defined and used for  deformation in the framework of warped convolutions, \cite{BS,BLS}.
\begin{definition} Let the one-particle Hilbert space be given as 
$\mathscr{H}_{1}:=L^2(d^n\mu(\mathbf{p}),\mathbb{R}^n)=\{f : \int 
d^n\mu(\mathbf{p})  |f(\mathbf{p})|^2<\infty, d^n\mu(\mathbf{p}):= 
 (2\omega_{\mathbf{p}} )^{-1}\,d^n\mathbf{p}, (\omega_{\mathbf{p}},\mathbf{p})\in H_m^{+}:=\{p\in\mathbb{R}^{d}|p^2=m^2, p_0>0 \} 
\}$ for $d-1=n\geq 1$ and let $\Delta (P)$ be the 
dense domain of all functions from $ \mathscr{H}_{1} $ vanishing at 
infinity faster than any inverse polynomial in $p^k$ given as follows, \cite[Equation III.24]{SV}
\begin{equation}
\Delta  (P) =\{f\in \mathscr{H}_{1} :
|\left(\mathbf{p}^2
\right)^rf(\mathbf{p})|\leq c_r(f)<\infty;\quad r=0,1,2,\dots\}.
\end{equation}
 $\Delta (P)$ is contained in the domain of the essential self-adjoint momentum operators. The
extended dense domain of the second quantized momentum operator $P_{\mu}$  is given by $\Delta_k(P):=\bigotimes_{i=1}^{k} 
\Delta(P)$ (for details concerning second-quantization see \cite[Theorem VIII.33]{RS1} and \cite[Example 2]{RS1}). 
\end{definition}
\begin{definition} \label{pux1}
Let the operator $X_{\mu}$ be defined by a unitary
equivalence to the momentum operator as follows,
\begin{equation}\label{pux}
 X_{\mu}=\Gamma(V^{-1})P_{\mu}\Gamma(V),
\end{equation}
where the operator $\Gamma(V):=\bigoplus\limits_{k=0}^{\infty} V^{\otimes k}$ 
is the second quantization of a unitary operator $V:\mathscr{H}_{1}\rightarrow \mathscr{H}_{1} $ which may depend on several real parameters.
\end{definition}

\begin{proposition}
 The operator $X_{\mu}$ defined by unitary equivalence (see Definition \ref{pux1}) is an
essentially self-adjoint operator on the dense
domain 
\begin{equation} \Delta_k\,(X):=
 \Gamma(V^{-1})\Delta_k\, (P),
\end{equation}
  commuting along its components, i.e. 
\begin{equation*}
 [X_{\mu},X_{\nu}]=0.
\end{equation*}
Therefore, the following  operator
\begin{equation}\label{uopx}
V(p)=e^{ip_{\mu}X^{\mu}},
\end{equation}
 is unitary and defines a strongly continuous group for all $p\in\mathbb{R}^d$.
\end{proposition}
\begin{proof}
By unitary equivalence essential self-adjointness of the operator $X_{\mu}$ follows.  The density of the domain $\Delta_k\,(X)$ follows from the density of the unitary equivalent domain  $\Delta_k\,(P)$  (see \cite[Lemma 2]{SV}).
In order to show the commutation relations between the different components of
the operator $X_{\mu}$ we use the unitary equivalence to the commuting momentum operators.
\begin{align*}
  [X_{\mu},X_{\nu}]&=[\Gamma(V^{-1})P_{\mu}\Gamma(V),
\Gamma(V^{-1})P_{\nu}\Gamma(V)]\\&=
\Gamma(V^{-1})[P_{\mu} ,
 P_{\nu}]\Gamma(V)\\&=0,
\end{align*}
where in the last line we used the fact that the momentum operator commutes along its components. Since $X$ is a self-adjoint operator it follows from (\cite[Theorem IIIX.7]{RS1}) that $V(p)$ defines a strongly continuous unitary
group.
\end{proof}

\begin{definition}  \label{defqfx}
  Let $\theta$ be a real skew-symmetric matrix w.r.t. the Lorentzian scalar-product on
$\mathbb{R}^{d}$ and  
let $\chi \in\mathscr{S}(\mathbb{R}^d\times\mathbb{R}^d)$ with $\chi(0,0)=1$.  Furthermore, let
$\phi(f)$ be the massive free scalar 
field smeared out with functions $f \in \mathscr{S}(\mathbb{R}^d)$. Then,
the operator valued distribution $\phi(f)$ deformed with the    operator $X_{\mu}$
(see Definition \ref{pux1}), 
denoted as $\phi_{\theta,X}(f)$,  is defined on vectors of the dense domain 
$\Delta_k\,(X)$  as follows
\begin{align}\label{defqfx2}
\phi_{\theta,X}(f)\Psi_{k}:&=(2\pi)^{-d}
\lim_{\epsilon\rightarrow 0}
  \iint  \, dx \,  dy \, e^{-ixy}  \, \chi(\epsilon x,\epsilon y)\beta_{{\theta x}
}(\phi(f))V({y})\Psi_{k}\nonumber\\&=
(2\pi)^{-d}\lim_{\epsilon\rightarrow 0}
  \iint  \, dx \,  dy \, e^{-ixy}  \, \chi(\epsilon x,\epsilon y) \beta_{{\theta x}
}\left(a(\overline{f^-})+a^*(f^+)\right)V({y})\Psi_{k}\nonumber\\&
=:\left(a_{\theta,X}(\overline{f^-})+a_{\theta,X}^*(f^+)\right)\Psi_{k}.
\end{align}
The automorphism $\beta$ is defined by the adjoint action of the unitary operator $V(y)$ and the
test functions $f^{\pm}(\mathbf{p})$ in momentum space are defined as follows
\begin{equation}\label{tf}
f^{\pm}(\mathbf{p}):=\int dx \,f(x)e^{\pm ipx}, \qquad 
p=(\omega_{\mathbf{p}},\mathbf{p})\in H_m^{+}.
\end{equation}
\end{definition}
 The integral (\ref{defqfx2}) has to be understood as an integral in oscillatory sense,
\cite{RI}. The unboundedness of the operator $X_{\mu}$ questions the existence of the
integral since we are dealing with unbounded operator valued distributions. To
show that the integral (\ref{defqfx2}) converges   we use the unitary
equivalence. \newline\newline
The following lemma proves the existence of a
unitary transformation connecting the warped convolutions of a free  scalar field using
the momentum operator, and the warped convolutions of a free scalar field using the  
unitary equivalent operator $X$.\newline
\begin{lemma} \label{ldx1}
  For $f \in \mathscr{S}(\mathbb{R}^d)$ and $\Psi_k \in \Delta_k\,(X)$, a 
transformation exists that maps the field deformed with the momentum 
operator $\phi_{\theta,{P}}(f)$ to the field  deformed with  operator $X$, i.e. 
$\phi_{\theta,{X}}(f)$. This transformation is given as 
follows
\begin{equation}\label{ldx1e}
  \phi_{\theta,{X}}(f)\Psi_k=
 \Gamma(V^{-1})\,\phi(Vf)_{
\theta,{P}}\,\Gamma(V)\Psi_k.
\end{equation}
\end{lemma}
\begin{proof} 
By using the unitary equivalence given in Equation (\ref{pux}), the lemma is easily proven
\begin{align*}
  \phi_{\theta,{X}}(f)\Psi_k&=(2\pi)^{-d}\lim_{\epsilon\rightarrow 0}
  \iint  \, dx \,  dy \, e^{-ixy}  \, \chi(\epsilon x,\epsilon y) \,V( {\theta
x})\phi(f)V(- {\theta
x}+ {y})\Psi_{k}\\
&=(2\pi)^{-d} \lim_{\epsilon\rightarrow 0}
  \iint  \, dx \,  dy \, e^{-ixy}  \, \chi(\epsilon x,\epsilon y)\,\Gamma(V^{-1}) U(
{\theta
x})\Gamma(V)  \phi(f)
\Gamma(V^{-1})\\&\hspace{8cm} \times U(- {\theta x}+ {y})\Gamma(V) 
\Psi_{k}
\\
&=
\Gamma(V^{-1})
\left(\Gamma(V)\phi(f)\Gamma(V^{-1})\right)_{
\theta, {P}}\Gamma(V) \Psi_k.
\end{align*}
\end{proof}
\begin{lemma} \label{ldx2}
  For  $\Phi_k \in\Delta_k\,(X) $ the familiar 
bounds of  the free field hold for the deformed field 
$\phi_{\theta, {X}}(f)$ and therefore the deformation with  operator $X_{\mu}$
 is well-defined.
\end{lemma} 
\begin{proof} 
By using Lemma \ref{ldx1} one obtains the familiar bounds for a free 
scalar field. For $\Phi_k \in\Delta_k\,(X)$ there exists a $\Psi_k \in \Delta_k\,(P)$ such
that the following holds
\begin{align*} \left\Vert\phi_{\theta, {X}}(f)\Phi_k\right\Vert&=
\left\Vert\phi_{\theta, {X}}(f)\Gamma(V^{-1}) \Psi_k\right\Vert\\&= \left\Vert 
(\Gamma(V)\phi(f)\Gamma(V^{-1}))_{\theta, {P}}\Psi_k
\right\Vert 
=\left\Vert (\phi( Vf))_{\theta, {P}}\Psi_k \right\Vert 
\\&\leq 
\left\Vert ( a( \overline{Vf^-}))_{\theta, {P}}\Psi_k \right\Vert
+\left\Vert (a^*( Vf^+))_{\theta, {P}}\Psi_k \right\Vert \\& \leq 
\left\Vert  Vf^{+} \right\Vert  \left\Vert (N+1)^{1/2 }\Psi_k 
\right\Vert
+ \left\Vert  Vf^{-}\right\Vert   \left\Vert (N+1)^{1/2 
}\Psi_k\right\Vert\\&=
\left\Vert   f^{+} \right\Vert  \left\Vert (N+1)^{1/2 }\Psi_k 
\right\Vert
+ \left\Vert  f^{-}\right\Vert   \left\Vert (N+1)^{1/2 
}\Psi_k\right\Vert.
\end{align*}
where in the last lines we used the triangle inequality, the Cauchy-Schwarz inequality and the 
bounds given in \cite{GL1}. \newline\newline
The obtained bounds are exactly the bounds of the free scalar field. Thus by the bounds of the free field it follows that the
field deformed with the operator $X_{\mu}$ is well-defined.

\end{proof}
 \subsection{Wightman properties of the deformed QF}
It is important to note that due to the unitary equivalence we can show that the
deformed field $\phi_{\theta,X}$ 
satisfies the Wightman properties with the exception of covariance and 
locality. This is the subject of the following proposition.  We shall
use the symbol $\mathscr{H}$ for  Bosonic Fockspace.\newline
\begin{proposition}\label{prop1x}
  Let $\theta$ be a real skew-symmetric matrix w.r.t. the Lorentzian scalar-product on
$\mathbb{R}^{d}$ and $f \in \mathscr{S}(\mathbb{R}^d)$.\\
\begin{itemize}
     \item[a)] The dense subspace $\mathcal{D}$ of vectors of finite 
particle number is contained in the domain 
$\mathcal{D}^{\theta, {X}}=\{\Psi\in \mathscr{H} |  \left\Vert 
\phi_{\theta, {X}}(f)\Psi\right\Vert^2 < \infty \}$ of any 
$\phi_{\theta, {X}}(f)$. Moreover, 
$\phi_{\theta, {X}}(f)\mathcal{D}\subset\mathcal{D}$ and 
$\phi_{\theta, {X}}(f)\Omega=\phi(f)\Omega$.\\
     \item[b)] For scalar fields deformed via warped convolutions and 
$\Psi\in\mathcal{D}$,
\begin{equation*}
f\longmapsto\phi_{\theta, {X}}(f)\Psi
\end{equation*}
is a vector valued tempered distribution.\\
     \item[c)]
For $\Psi\in\mathcal{D}$  and $\phi_{\theta, {X}}(f)$  the following holds
\begin{equation*}
\phi_{\theta, {X}}(f)^{*}\Psi=\phi_{\theta, {X}}(\overline{f})\Psi.
\end{equation*} 
For real $f\in\mathscr{S}(\mathbb{R}^{d})$, the deformed field 
$\phi_{\theta, {X}}(f)$ is essentially self-adjoint on
$\mathcal{D}$.\\
     \item[d)] The Reeh-Schlieder property holds: Given an  open set of 
space-time $\mathcal{O}\subset \mathbb{R}^d$, then
\begin{equation*}
\mathcal{D}_{\theta, {X}}(\mathcal{O}):=  
span\{\phi_{\theta, {X}}(f_1)\dots\phi_{\theta, {X}}(f_k)\Omega: k\in \mathbb{N}, 
f_1\dots f_k\in \mathscr{S}(\mathcal{O})\}
\end{equation*}
is dense in $\mathscr{H}$.
  \end{itemize}
\end{proposition}
\begin{proof} 
a) The fact that $\mathcal{D}\subset \mathcal{D}^{\theta, {X}}$,  follows 
 immediately from Lemma \ref{ldx2} , since the deformed scalar field 
satisfies the same bounds as a free field. The fact that the 
deformed field acting on the vacuum is the same as the free field acting 
on $\Omega$, can be easily shown due to the property of the unitary 
operators $V( {b})\Omega=\Omega$ (see \cite[Chapter X.7]{RS2}). 
\\\\
b) By using Lemma \ref{ldx2} one can see that the right hand side 
depends continuously on the function $f$, hence the temperateness of 
$f\longmapsto\phi_{\theta, {X}}(f)\Psi$,  $\Psi\in\mathcal{D}$ follows.\\\\
c) First, we prove hermiticity of the deformed field  $\phi_{\theta, {X}}(f)$. This is done
   along the same lines as the proof of Lemma \ref{wcl2}, demonstrating hermiticity of a deformed
operator if the undeformed one is self-adjoint. 
   \begin{align*}
  \phi_{\theta, {X}}(f)^{*}\Psi&= (2\pi)^{-d}\left(\lim_{\epsilon\rightarrow 0}
  \iint  \, dx \,  dy \, e^{-ixy}  \, \chi(\epsilon x,\epsilon y)
  \,\beta_{\theta x}(\phi(f))V({y})\right)^{*}\Psi 
\\&=(2\pi)^{-d}
 \lim_{\epsilon\rightarrow 0}
  \iint  \, dx \,  dy \, e^{-ixy}  \, \overline{\chi(\epsilon x,-\epsilon y)}\,V({y})\beta_{
{\theta
x}}(\phi(f))^* \Psi
  \\&=(2\pi)^{-d}
 \lim_{\epsilon\rightarrow 0}
  \iint  \, dx \,  dy \, e^{-ixy}  \, \overline{\chi(\epsilon (x+\theta^{-1}y) ,-\epsilon
y)}\,\beta_{ {\theta
x}}(\phi(\overline{f }))V(y)\Psi
  \\&=
\phi_{\theta, {X}}(\overline{f })\Psi.
\end{align*}
In the last lines we performed a variable substitution $\left(y_{\mu}\rightarrow -y_{\mu}\right)$
and $\left(x_{\mu}\rightarrow 
x_{\mu}+(\theta^{-1}y)_{\mu}\right)$.
 \\\\ For real
$f$ we can prove the essential  self-adjointness of the hermitian deformed field $\phi_{\theta,
{X}}(f)$. The first step consists in showing that the
deformed field has a dense set of analytic vectors. Next, by  
Nelson's analytic vector theorem, it  follows that the deformed field
$\phi_{\theta, {X}}(f)$ is essentially self-adjoint on this dense set of analytic vectors, (for
similar proof see \cite[Chapter I, Proposition 5.2.3]{BR}). \newline\newline  For $\Psi_{k} \in
\mathscr{H}_{k}$ the estimates of the $l$-power
of the deformed field $\phi_{\theta, {X}}(f)$, are given in the following
\begin{equation*}
 \left\Vert
\phi_{\theta, {X}}(f)^l\Psi_{k}\right\Vert\leq
2^{l/2}(k+l)^{1/2}(k+l-1)^{1/2}\cdots(k+1)^{1/2}\left\Vert f\right\Vert^{l}\left\Vert \Psi_{k}
\right\Vert,
\end{equation*}
where in the last lines we used Lemma \ref{ldx2} for the estimates of the deformed field. Finally,
we can write the sum
 
\begin{equation*}
 \sum\limits_{l\geq 0} \frac{\vert t\vert^l}{l!}\left\Vert
\phi(f)^l \Psi_k
\right\Vert
\leq
 \sum\limits_{l \geq 0} \frac{(\sqrt{2}|t|)^l}{l!}\left(
\frac{(k+l)!}{k!}\right)^{1/2}\left\Vert f\right\Vert^{l}\left\Vert \Psi_{k} \right\Vert<
\infty
\end{equation*}
for all $t\in \mathbb{C}$. It follows that each $\Psi\in \mathcal{D}$ is an analytic vector for
the deformed field $\phi_{\theta, {X}}(f)$. Since the set $\mathcal{D}$ is dense in
$\mathscr{H} $,  Nelson's analytic vector theorem implies that $\phi_{\theta, {X}}(f)$
is essentially self-adjoint on $\mathcal{D}$. 
\\\\
d) For the proof of the Reeh-Schlieder property we use the 
unitary equivalence given in Definition (\ref{pux1}).
  First note that the spectral properties of the unitary operator  $V( {y})$, are the same
as for the   unitary operator  $U( {y})$
 of translations. This leads to the application of the 
standard Reeh-Schlieder argument \cite{SW} which states that that 
$\mathcal{D}_{\theta}(\mathcal{O})$ is dense in $\mathscr{H} $ if and 
only if $\mathcal{D}_{\theta}( {\mathbb{R}^d})$ is dense in 
$\mathscr{H} $.  We choose the functions  $f_1,\dots,f_k \in \mathscr 
{S}(\mathbb{R}^{d})$ such that the Fourier transforms of the functions 
do not intersect the lower mass shell and therefore the domain 
$\mathcal{D}_{\theta}({\mathbb{R}^d})$ consists of the following vectors
   \begin{align*} 
   \Gamma(V)\phi_{\theta, {X}}(f_1)\dots
\phi_{\theta, {X}}(f_k)\Omega&=
   \Gamma(V)a_{\theta, {X}}^*(f^{+}_1)\dots
a_{\theta, {X}}^*(f^{+}_k)\Omega 
\\&=   \Gamma(V)   \Gamma(V^{-1}) 
     a_{\theta,
 {P}}^*( Vf^{+}_1)  \dots  a_{\theta,
 {P}}^*( Vf^{+}_k) \Gamma(V^{-1})
\Omega\\&=
  a_{\theta,
 {P}}^*( Vf^{+}_1)  \dots  a_{\theta,
 {P}}^*( Vf^{+}_k) 
\Omega\\&=
  \sqrt{m!}P_{m}(S_m( Vf^{+}_1\otimes \dots \otimes V
f^{+}_k)),
\end{align*}
where $P_k$ denotes the orthogonal projection from 
$\mathscr{H}_1^{\otimes k}$ onto its totally symmetric subspace 
$\mathscr{H}_{k}$, and $S_k \in \mathscr{B}(\mathscr{H}_1^{\otimes 
k})$ is the multiplication operator given as
    \begin{align*} S_k(p_1,\dots,p_k)=\prod \limits_{1\leq l < j \leq k} 
e^{{i}p_l \theta p_j }.
\end{align*}
Since the operator $\Gamma(V)$ is unitary,  
functions  $Vf^{+}_k \in \mathscr {S}(\mathbb{R}^{d})$ for $f^{+}_k \in
\mathscr {S}(\mathbb{R}^{d})$ will give
rise to dense sets of functions in $\mathscr{H}_1$. Following the 
same arguments as in \cite{GL1} the density of $\mathcal{D}_{\theta}( 
{\mathbb{R}^d})$ in $\mathscr{H}$ follows. Note that we proved the 
density for vectors $   \Gamma(V)\phi_{\theta, {X}}(f_1)\dots 
\phi_{\theta, {X}}(f_k)\Omega$ and not for the vectors without the application of $  
\Gamma(V) $ as stated in the proposition. However, we use the unitarity of $  
\Gamma(V)$ to argue that vectors dense in $\mathscr{H}$ stay dense after 
the application of a unitary operator.
\end{proof}

\subsection{Wedge-Covariance and Wedge-Locality}\label{s32}
The authors in \cite{GL1} constructed a map
$Q:W \mapsto Q(W)$ from a 
set $\mathcal{W}_{0}:=\mathcal{L}^{\uparrow}_{+}W_{1}$ of wedges, where 
$W_{1}:=\{x\in \mathbb {R}^d: x_1 > |x_0| \}$  to a set 
$\mathcal{Q}_{0}\subset \mathbb{R}^{-}_{d\times d}$ of skew-symmetric 
matrices. In the next step they considered the corresponding fields 
$\phi_{W}(x):=\phi(Q(W),x)$. \\\\ 
Hence, the correspondence is understood as a scalar field $ \phi(Q(W),x)$  on a NC space-time, 
which can be equivalently realized as a field defined on the wedge. The homomorphism $Q:W \mapsto Q(W)$   is given by the following definitions.

\begin{definition}Let  $\theta$ be a real skew-symmetric matrix on 
$\mathbb{R}^d$ then the map $\gamma_{\Lambda}(\theta)$
is defined as follows
\begin{align}\label{hm}
  \gamma_{\Lambda}(\theta):=
\left\{
\begin{array} {cc}
\Lambda\theta\Lambda^T, \qquad &\Lambda\in\mathcal{L}^{\uparrow}, \\ 
-\Lambda\theta\Lambda^T,\qquad &\Lambda\in\mathcal{L}^{\downarrow} .
\end{array} \right.
\end{align}
\end{definition}
\begin{definition}
$\theta$ is called an admissible matrix if the realization of the 
homomorphism $Q(\Lambda W)$ defined by the map 
$\gamma_{\Lambda}(\theta)$   is a well defined mapping. This is the case 
iff $\theta$ has in $d$ dimensions the following form
\begin{align}
\begin{pmatrix} 0 & \lambda & 0 & \cdots & 0 \\ \lambda& 0 & 0 & \cdots 
& 0  \\ 0 & 0 &0 &\cdots&0
\\ \vdots & \vdots &\vdots & \ddots& \vdots\\ 0 & 0 &0 & \cdots&0
  \end{pmatrix},\qquad \lambda\geq0.
\end{align}
For the physical most interesting case of 4 dimensions the 
skew-symmetric matrix $\theta$ has the more general form
\begin{align}\qquad\quad
\begin{pmatrix} 0 & \lambda & 0 & 0 \\ \lambda& 0 & 0 & 0  \\ 0 & 0 &0 &\eta
  \\ 0 & 0 &-\eta &0
  \end{pmatrix},\qquad \lambda\geq0, \eta \in \mathbb{R}.
\end{align}

\end{definition} 
By using the former definitions we give the following correspondence of the fields,
\begin{equation}\label{ex1}
\phi_W(f):=\phi(Q(W),f)=\phi_{\theta, {X}}( f). 
\end{equation}Next, we turn our attention to the covariance and locality of the defined fields. Wedge-covariance and -locality  seems to be the appropriate locality on non-commutative space-times,  \cite{Sol}.  In the following we  lay out the definitions of a wedge-covariant and a wedge-local field, (\cite{GL1}, Definition 3.2).
\newline
 \begin{definition} \label{drca}
Let $\phi=\{\phi_W: W\in\mathcal{W}_{0}\}$ denote the family of fields
satisfying the domain and continuity assumptions of the Wightman axioms. Then, the field  $\phi $
is defined to be a wedge-local quantum field if the following  
conditions are satisfied: 
\\ 
\begin{itemize}
      \item \textbf{Covariance:}  For any $W \in\mathcal{W}_{0} $ and 
$f\in\mathscr{S}(\mathbb{R}^d)$ the following holds
\begin{align*}
      U( y, \Lambda)\phi_W(f)U( y, \Lambda)^{-1}&=\phi_{\Lambda W}(f\circ( y,
\Lambda)^{-1} ),\qquad  (y, \Lambda)\in \mathcal{P}^{\uparrow}_{+},
\\
      U(0, j)\phi_W(f) U( 0,j)^{-1}&=\phi_{ jW}(\overline{f}\circ (0, j)^{-1}) ,
\end{align*}$\,$  where $j$ represents the space-time reflections, i.e. $x^{\mu}\rightarrow -x^{\mu}.$\\
      \item  \textbf{Wedge-locality:}  Let $W,\tilde{W}\in\mathcal{W}_{0}$ and
$f\in\mathscr{S}(\mathbb{R}^2)$. If 
\begin{equation*}
\overline{W+\text{supp } f}\subset
(\tilde{W}+\text{supp } g)',
\end{equation*}  \end{itemize}
then 
\begin{equation*}
 [\phi_{W}(f), \phi_{\tilde{W}}(g)]\Psi=0,\quad \Psi \in \mathcal{D}.
\end{equation*}

\end{definition}
The last definition can be given in a simpler form due to the geometrical properties of the
wedges. This is the subject of the following lemma, (\cite{GL1}, Lemma 3.3).\newline
\begin{lemma}
 Let $\phi=\{\phi_W: W\in\mathcal{W}_{0}\}$ denote the family of fields
satisfying the domain, continuity and covariance assumptions stated in Definition \ref{drca}. Then $\phi$ is
wedge-local if and only if
\begin{equation*}
 [\phi_{W_1}(f), \phi_{-{W}_1}(g)]\Psi=0,\quad \Psi \in \mathcal{D},
\end{equation*}
for all $f,g \in C_{0}^{\infty}(\mathbb{R}^d)$ with  $\text{supp } f\subset W_{1}$ and 
$\text{supp } g\subset -W_{1}$.
\end{lemma}$\,$ \\So let us first investigate the wedge-covariance properties of our deformed fields. The result is given in the following proposition. 
\begin{proposition}\label{wl0} The deformed fields $\phi_{\theta, X}(    f)  $ transform  under the adjoint action of the proper orthochronous Poincar\'e group as follows,
\begin{align*} 
 U( x, \Lambda) \phi_{\theta, X}(    f)   U( x, \Lambda)^{-1}=\phi_{\theta, U( x, \Lambda)XU( x, \Lambda)^{-1}} (    f\circ  ( x,
\Lambda)^{-1}). \end{align*}
  Let the operator $X$ be covariant w.r.t.  the  proper orthochronous Lorentz group. Then,  the field is wedge-covariant w.r.t.  the  proper orthochronous Lorentz group
i.e. 
\begin{align*}  U( 0, \Lambda) \phi_{\theta, X}(    f)   U( 0, \Lambda)^{-1}=\phi_{\gamma_{\Lambda}(\theta) , X} (    f\circ  ( 0,
\Lambda)^{-1}) .
   \end{align*}
Moreover, if the operator $X$ is covariant w.r.t.  the  proper orthochronous Poincar\'e group and  the space-time reflections, then the field $\phi$ is a  wedge-covariant field. 
 \end{proposition}
 
\begin{proof} \begin{align*} &
 U( x, \Lambda) \phi_{\theta, X}(    f)   U( x, \Lambda)^{-1}\\&=(2\pi)^{-d} 
\lim_{\epsilon\rightarrow 0}
  \iint  \, dy \,  du \, e^{-iyu}  \, \chi(\epsilon y,\epsilon u)  U( x, \Lambda)\beta_{{\theta y}
}(\phi(f))V({u}) \,  U( x, \Lambda)^{-1}\\&=(2\pi)^{-d} 
\lim_{\epsilon\rightarrow 0}
  \iint  \, dy \,  du \, e^{-iyu}  \, \chi(\epsilon y,\epsilon u) V_{\Lambda, x}({\theta y}) 
  \,\phi(f\circ( x,
\Lambda)^{-1})\,V_{\Lambda, x}(-\theta y+ u ) \,   
\\&=  \phi_{\theta, U(\Lambda,x)XU(\Lambda,x)^{-1}} (    f\circ  ( x,
\Lambda)^{-1}) ,
\end{align*}
where $V_{\Lambda, x}(y):= U( x, \Lambda)V(y) U( x, \Lambda)^{-1}=e^{iy U( x, \Lambda)XU( x, \Lambda)^{-1}}$. 
 Now this expression is nothing else than the operator $X$ used for deformation but unitarly transformed. The second and third part follow from Proposition \ref{blsp1}, where   in the case of space-time reflections one replaces the smearing function $f$ with $\overline{f}$.\end{proof}
\begin{remark}
 An operator that is translation invariant is not equivalent to the momentum operator, for example 
\begin{align*}
 X=U( \Lambda)P U( \Lambda)^{-1} ,    \qquad X=e^{iaD}Pe^{-iaD} ,\end{align*}
 where $D$ is the dilatation operator  which is only essentially self-adjoint in the massless case, \cite{W3}. 
\end{remark}
What information do we gain from the former proposition? It gives us the transformational behavior of a field defined on a wedge that can be associated with a excitation on  a non-commutative space-time.  Under the  assumption of Lorentz covariance for the operator,  it states that the field obtained by a Poincar\'{e} transformation associates to a transformed field generated by deformation with $U( x)XU( x)^{-1}$. The interpretation of the result is the following.
Since the deformed fields generated by $X$ are associated to a non-commutative space-time, fields generated by  $U( x)XU( x)^{-1}$ correspond to fields on an equivalent but  translated quantum space-time. Hence, we  already are able to  deduce from this result that generators of constant quantum space-times shall be translationally invariant. This will be further studied in Section \ref{s5} where we examine the isomorphism to non-commutative space-times.  
\newline\newline
Next we turn to the original proof of wedge-locality. It is usually done by showing that the functions used to smear the field are entire analytic  and therefore they can be analytically continued to the complex upper half plane. The proof is done by introducing suitable coordinates given by,
\begin{equation*}
m_{\perp}:=(m^2+p_{\perp})^{1/2}, \qquad p_{\perp}:=(p_2,\dots,p_{n}),\qquad \vartheta:=\arcsinh\frac{p_1}{m_{\perp}}.
\end{equation*}
In the new coordinates we have the following measure and  on-shell momentum vector,
\begin{equation*}
d^n\mu(\mathbf{p})= d^{n-1}p_{\perp}d\vartheta,\qquad \qquad p(\vartheta):=\left(
\begin{array}{c}
m_{\perp}\cosh \vartheta \\
m_{\perp}\sinh \vartheta  \\
p_{\perp}
\\
\end{array}
\right)
\end{equation*}
By using these new coordinates, the analyticity of the function and the analytic continuation one obtains for the smeared functions $f\in
C_{0}^{\infty}(W_{1})$ and $g\in C_{0}^{\infty}(-W_{1})$, (see \ref{tf})
\begin{equation}\label{ac1}
 f^{-}(p_{\perp}, \vartheta+i\pi)=
 f^{+}(-p_{\perp},\vartheta) ,\qquad  g^{-}(p_{\perp},  \vartheta+i\pi)=
 g^{+}(-p_{\perp},\vartheta) .
\end{equation}
Now for the proof of wedge-locality   we have to demand that the unitary transformed functions, i.e. $Vf^{-}(p_{\perp},  \vartheta)$ and $V g^{+}(p_{\perp},  \vartheta)$ satisfy the demanded analyticity and analytical continuation properties.  Note that this  restrains  the unitary operators  used in the definition of the operator that is unitary equivalent to the momentum operator. By taking the former definition and lemma into account the following proposition concerning the deformed field
$\phi_{\theta, X}$ follows. \newline
\begin{proposition}\label{wl1}Let the unitary transformation $V$  leave the support  for all $f,g \in C_{0}^{\infty}(\mathbb{R}^d)$ with  $\text{supp } f\subset W_{1}$ and 
$\text{supp } g\subset -W_{1}$ covariant, i.e.
\begin{equation*}
  \text{supp } V f\subset W_{1}  \qquad
 \text{supp }V g\subset -W_{1}  .
\end{equation*}
Then, the family of fields $\phi=\{\phi_W: 
W\in\mathcal{W}_{0}\}$ defined by $\phi_W(f):=\phi(Q(W),f)=\phi_{\theta, {X}}( f)=\phi_{\theta, {V^{-1}PV}}( f)$
are wedge-local fields on the Bosonic Fockspace $\mathscr{H}^{+}$.
\end{proposition}

\begin{proof}
For the proof  we use Proposition \ref{pwlf}, the unitary equivalence given in Lemma \ref{ldx1}
and the proof that the free scalar field deformed with the momentum operator is wedge local,
\cite{GL1}.
To use Proposition \ref{pwlf}, we have to show that the following commutator vanishes for $f\in
C_{0}^{\infty}(W_{1})$ and $g\in C_{0}^{\infty}(-W_{1})$,
\begin{align*}
[\beta_{\theta x}(\phi(f)),\beta_{-\theta y}(\phi(g))]&=
[\beta_{\theta x}(a(\overline{f^-})),\beta_{-\theta y}(a^*({g^+})]-
[\beta_{-\theta y}(a(\overline{g^-})),\beta_{\theta x}(a^*({f^+})]\\&=\Gamma(V^{-1})
[\alpha_{\theta x}\bigl( 
\phi(V f) 
\bigr),\alpha_{-\theta y}\bigl(
 \phi(V g) 
\bigr)
 ] \Gamma(V)
,
\end{align*}
where in the former lines all other terms are equal to zero and the unitary equivalence was used.  
 Let us first take a look at the first expression of the
commutator, 
\begin{align*} 
 &
[\alpha_{\theta x}\bigl(
\Gamma(V)
a(\overline{f^-})
\Gamma(V^{-1})
\bigr),\alpha_{-\theta y}\bigl(
\Gamma(V)
a^*({g^+})
\Gamma(V^{-1})
\bigr)
 ] 
\\&
=
\int d^n\mu(\mathbf{p})\int d^n\mu(\mathbf{k})
(Vf^{-})(\mathbf{p})(Vg^{+})(\mathbf{k}) e^{-ip\theta x} e^{-ik\theta y} 
[  a (\mathbf{p}) ,  a^* (\mathbf{k}) ] 
   \\& 
=
\int d^n\mu(\mathbf{p}) 
(Vf^{-})(\mathbf{p})(Vg^{+})(\mathbf{p}) e^{-ip\theta (x+y)}\\&
=
\int  d^{n-1}p_{\perp}d\vartheta\,
(Vf^{-})(p_{\perp},\vartheta)(Vg^{+})(p_{\perp},\vartheta)e^{-ip(\vartheta)\theta (x+y)}\\&
=
\int  d^{n-1}p_{\perp}d\vartheta\,
(Vf^{+})(-p_{\perp},\vartheta)(Vg^{-})(-p_{\perp},\vartheta)e^{-ip(\vartheta+i\pi)\theta (x+y)}\\&
=
\int d^n\mu(\mathbf{p}) 
 (Vf^{+})(\mathbf{p})(Vg^{-})(\mathbf{p}) e^{ ip\theta (x+y)},
\end{align*}
where in the last lines we used the unitary equivalence (\ref{pux}),  the boundedness and analyticity properties of the unitary transformed functions $f,g$ (see \cite[Proposition 3.4]{GL1}) and we shifted the contour of the integral from $\mathbb{R}$ to $\mathbb{R}+i\pi$.  Next, we look at the second expression of the
commutator and obtain the following, 
\begin{align*}&[
\alpha_{-\theta y}\bigl(
\Gamma(V)
a (\overline{g^-})
\Gamma(V^{-1})
\bigr),\alpha_{\theta x}\bigl(
\Gamma(V)
a^*( {f^+})
\Gamma(V^{-1})
\bigr)] \\&
=
\iint d^n\mu(\mathbf{p}) d^n\mu(\mathbf{k})
 (Vf^{+})(\mathbf{p})(Vg^{-})(\mathbf{k}) e^{ip\theta x} e^{ik\theta y}
[ a (\mathbf{k}), a^* (\mathbf{p})]
   \\&
=
\int d^n\mu(\mathbf{p}) 
(Vf^{+})(\mathbf{p})(Vg^{-})(\mathbf{p}) e^{ ip\theta (x+y)} .
\end{align*}
Since the second expression of the commutator $
[\beta_{\theta x}(\phi(f)),\beta_{-\theta y}(\phi(g))] $ is equal to the first one with a sign
difference, the commutator vanishes. Hence, the fields $\phi_{W}$ are wedge-local. 
\end{proof}$\,$\\
Concerning the wedge-covariance we imposed a strong requirement on the choice of our unitary operators. In particular, we demanded the unitary transformation $V$
to  leave the support  for all $f,g \in C_{0}^{\infty}(\mathbb{R}^d)$ with  $\text{supp } f\subset W_{1}$ and 
$\text{supp } g\subset -W_{1}$ covariant. Are there any examples of such transformations? This question will be answered positively by introducing a few examples.
\begin{example}
We first mention the Lorentz-transformation, i.e. $\Gamma(V)=U(\Lambda)$ with $\Lambda \in \text{SO}(1,1)\times \text{SO}(d-2)$.
\begin{proof}
 For the right wedge, i.e. $x_1>|x_0|$ we have to verify the following inequality 
\begin{equation*} 
(\Lambda x)_1>|(\Lambda x)_0|.
\end{equation*}
This can be easily verified by using the property of the  wedge and the explicit form of the Lorentz boost in $0-1$ direction,
\begin{align*} 
-\gamma \beta x_0+\gamma x_1&>|\gamma x_0-\gamma\beta x_1|\\ 
-  \beta x_0+  x_1&>|\  x_0- \beta x_1|,
\end{align*}since the Lorentz-factor $\gamma>0$, moreover
\begin{align*}  
-  \beta x_0+  x_1&> -  \beta x_0+|x_0|>0, 
\end{align*}since the velocity coefficient $ |\beta|<1$. Thus we obtain 
\begin{align*}  
(-  \beta x_0+  x_1)^2&>(x_0- \beta x_1)^2\\
x_1^2(1-\beta^2)&>x_0^2(1-\beta^2)\\
x_1&>|x_0|.
\end{align*}
The proof for the left wedge is analogous. 
\end{proof}
This example is well known and intuitively easy to understand, since the group $\text{SO}(1,1)\times \text{SO}(d-2)\subset\mathcal{L}_{+}^{\uparrow}$ is the stabilizer group $\mathcal{L}^{\uparrow}_{+}(W_1, ) \subset \mathcal{L}_{+}^{\uparrow}$ of $W_1$.   
\end{example}
\begin{example}
Second we mention the unitary operator of translations in the momentum space $\Gamma(V(\vec{k})) =e^{i\vec{k}\cdot\vec{X}}$, where $\vec{X}$ is the second-quantized Newton-Wigner-Pryce operator. It was studied thoroughly in a QFT-context in \cite{Muc2, Muc3}. In particular the operator acts on the particle operators as follows,
\begin{align*}  
\Gamma(V(\vec{k}))a(\mathbf{p})\Gamma(V(\vec{k})^{-1})=a(\mathbf{p}-\mathbf{k}),\qquad \Gamma(V(\vec{k}))a^{*}(\mathbf{p})\Gamma(V(\vec{k})^{-1})=a^{*}(\mathbf{p}-\mathbf{k}).
\end{align*}
Since the coordinate space remains invariant under such a transformation in the momentum space, the momentum-translation  $V(\vec{k})$ is our second most prominent example. 
\end{example}

\begin{example}
Note that the translation operator $U(y)=e^{iy_{\mu}P^{\mu}}$ for $y\in W_{1}$ leaves the support of $f\in C_{0}^{\infty}(\mathbb{R}^d)$ covariant (see \cite{BLS}). By using the former examples  we can take arbitrary arrangements of the operators $U(y)$, $\Gamma(V(\vec{k}) )$ and $U(\Lambda)$  and thus obtain a whole class of wedge-local fields. 
\end{example}

 \begin{example}Although we intend to focus on the massless scalar field in a forthcoming work, we mention in this context the special conformal transformation.  This operator leaves the wedge covariant and was intensively studied in \cite{MUc}. The operator $\Gamma(V)$ that gives the unitary equivalence to the momentum operator is in the special conformal case the inversion operator constructed  by \cite{SV}.
\end{example}

 \begin{example}Another interesting unitary operator that should be mentioned  in the massless case is given by the dilation operator, i.e. $\Gamma(V)=e^{ib D}$. It leaves the wedge covariant, since it represents merely a scale transformation  and the operator $ \Gamma(V^{-1}) P \Gamma(V) $ transforms covariantly under Poincar\'e transformations.
\end{example}

 The reader should be aware of the fact that wedge-covariance was not shown for the field $\phi_{\theta,X}$ although it is obligatory  when proving wedge-locality. 
Nevertheless, by reducing the proof of wedge-locality  for the field $\phi_{\theta,X}(f)$ to the field  $\phi_{\theta,P}(Vf)$, we were able to circumvent this particular problem.

\section{Scattering} 
The next task of this work is to calculate the Scattering-matrix  by using   \textbf{tempered polarization free generators}, \cite{Sch97,BBS}. 
In \cite{BBS} a framework was developed to calculate two-particle scattering of such a given theory, where the construction relies on the Haag-Ruelle scattering theory. In order to proceed let us briefly lay out the necessary definitions and properties.

\begin{definition}\label{pfg}
Let $W\in\mathcal{W}_{0}$ and $f\in\mathscr{S}(\mathbb{R}^d)$. Then the following properties constitute a tempered polarization free generator $\phi_{W}(f)$,
\\ 
\begin{itemize}
   \item[a)]   $\phi_{W}(f)$ is a wedge-local field.\\
     \item[b)]$\phi_{W}(f)$ and  $\phi_{W}(f)^*$ are closed operators with $\Omega$ contained in their respective domains. 
     \\ \item[c)]   $\phi_{W}(f)\Omega$ and $\phi_{W}(f)^*\Omega$ are single particle states.
 \\ \item[d)] $\phi_{W}(f)$  is said to be \textbf{temperate} if there is a dense subspace $\mathcal{D}$ of its domain which is stable under translation, such that
\begin{equation}
x\mapsto\phi_{W}(f) U(x)\Psi, \qquad\forall \Psi\in \mathcal{D}
\end{equation}
is strongly continuous and polynomially bounded in the norm for large $x$.
\end{itemize}
\end{definition}
\begin{lemma}
Let $W\in\mathcal{W}_{0}$ and $f\in\mathscr{S}(\mathbb{R}^d)$ and let the unitary operator $\Gamma( V )$  be as demanded in Proposition \ref{wl1}. Then, the set of fields  $\phi(Q(W),f)=\phi_{\theta, {X}}( f)$ constitute the properties of tempered polarization free generators. 
\end{lemma}

\begin{proof}
The first item  in Definition \ref{pfg} is wedge-locality. This follows easily from the choice of the unitary operator $\Gamma( V )$  (see Proposition \ref{wl1}).
Item  b) holds since,  $\phi_{\theta, {X}}( f)$ is a densely defined and symmetric operator, hence closeable. To see that the vacuum vector is contained in the domain see  Proposition \ref{prop1x}, a). It is straightforward to prove that the deformed fields generate single particle states and this property follows as well from   Proposition \ref{prop1x}, a). \\\\ Now let us turn our attention to the hardest part of this proof, the temperateness. Proving   continuity for the expression $\phi_{\theta, X}(f) U(x)\Psi$ is equivalent to proving it for $\phi_{\theta, P}(V f)  \Gamma(V)  U(x)\Psi$,
\begin{align*}&
||  \phi_{\theta, P}(  V f)  \Gamma(V)  U(x)\Psi-\phi_{\theta, P}( V f)  \Gamma(V)   \Psi||=||  \phi_{\theta, P}( V f)  \Gamma(V)  (U(x)-1)\Psi ||\\&\leq 
\left\Vert V  f^{+} \right\Vert  \left\Vert (N+1)^{1/2 }  \Gamma(V)  (U(x)-1)\Psi  
\right\Vert
+ \left\Vert V f^{-}\right\Vert   \left\Vert (N+1)^{1/2 
}  \Gamma(V) (U(x)-1)\Psi \right\Vert\\& =
\left\Vert   f^{+} \right\Vert  \left\Vert (N+1)^{1/2 }  (U(x)-1)\Psi  
\right\Vert
+ \left\Vert  f^{-}\right\Vert   \left\Vert (N+1)^{1/2 
}   (U(x)-1)\Psi \right\Vert
\overset{x\rightarrow 0\,\,\,\,\,\,} {\xrightarrow{\hspace*{0.8cm}} 0},
\end{align*}
where in the last lines we used the fact that $\Psi \in \mathcal{D}$ and by applying  unitary operators that do not change the particle number on vectors of finite particle number, we have $ \Gamma(V)  U(x) \Psi \in \mathcal{D}$ and hence we can use the   bounds given in Lemma \ref{ldx2}. Moreover, in the last expression we use  the strong continuity of $U$ for  $ \Psi \in \mathcal{D}$ and thus the scalar product and the limit can be interchanged. \\\\
Of course, the boundedness can be proven by following similar arguments as for the continuity. Nevertheless, a more elegant route is chosen, i.e. 
\begin{align*} &
||  \phi_{\theta, X}(    f)   U(x)\Psi||=(2\pi)^{-d}||  
\lim_{\epsilon\rightarrow 0}
  \iint  \, dy \,  du \, e^{-iyu}  \, \chi(\epsilon y,\epsilon u)\beta_{{\theta y}
}(\phi(f))V({u}) \, U(x)\Psi||\\&=(2\pi)^{-d}||  
\lim_{\epsilon\rightarrow 0}
  \iint  \, dy \,  du \, e^{-iyu}  \, \chi(\epsilon y,\epsilon u) V_{-x}({\theta y}) 
  \,\phi(f\circ (-x))\,V_{-x}(-\theta y+ u ) \, \Psi||
\\&= || \phi_{\theta,U(x)^{-1}XU(x)} (    f\circ ( x))      \Psi|| \\ &\leq 
\left\Vert  V(f^{+}\circ ( x)) \right\Vert  \left\Vert (N+1)^{1/2 }\Psi
\right\Vert
+ \left\Vert  V(f^{-}\circ ( x))\right\Vert   \left\Vert (N+1)^{1/2 
}\Psi\right\Vert\\ &=
\left\Vert   f^{+} \right\Vert  \left\Vert (N+1)^{1/2 }\Psi
\right\Vert
+ \left\Vert  f^{-}\right\Vert   \left\Vert (N+1)^{1/2 
}\Psi\right\Vert. 
\end{align*}
where $V_{-x}(y) = U(x)^{-1}V(y) U(x)=e^{iy U(x)^{-1}XU(x)}$. Now this expression is nothing else than the operator $X$ used for deformation but unitarly transformed. Hence, we simply have another operator that is unitary equivalent to the momentum operator. By using the bounds in  Lemma \ref{ldx2}, boundedness follows.
\end{proof}$\,$\\
Next, let us define a function for $t\in\mathbb{R}$ and $f\in \mathscr{S}(\mathbb{R}^d)$ by, 

\begin{equation*}
f_t(x)=(2\pi)^{-d/2}\int dp \,\widetilde{f}(p)\,e^{ipx}e^{i(p_0-\omega_{\mathbf{p}})t}.
\end{equation*}
The support properties of the functions $f_t$ for asymptotic $t$ are used in the subsequent discussion. To proceed, let us define the velocity support of $f$ by, 

\begin{equation*} 
\Xi(f)=\{(1,\mathbf{p}/ \omega_{\mathbf{p}}): p\in \text{supp}\,\widetilde{f}\}.
\end{equation*}
It follows that the support of $f_t$ is contained in $t\,\Xi(f)$. 
Furthermore, the partial ordering of the sets with reference to the wedge $\mathcal{W}_0$ have to be introduced.
\begin{definition}
Let $\Xi_a$, $\Xi_b$ $\subset\mathbb{R}^d$ be compact sets. $\Xi_a$ is said to be the precursor of $\Xi_b$, $\Xi_a$ $\prec$ $\Xi_b$ in formula form, if $\Xi_a- \Xi_b$ is contained in $W\in\mathcal{W}_0$.
\end{definition}$\,$\\
By using the former definitions and sophisticated techniques the authors were able to show that $
\phi_{W}(f_t)\phi_{W'}(g_t)\Omega$, 
converges to the incoming respectively outgoing two-particle states for $t\rightarrow \pm \infty$. For the test functions $f,g$ with disjoint momentum supports in a small neighborhood of some point on the mass shell one obtains,
\begin{align*}  
\lim_{t\rightarrow   \infty}\phi_{W}(f_t)\phi_{W'}(g_t)\Omega&=\left(\phi_{W}(f )\Omega\times \phi_{W'}(g )\Omega\right)_{out} \qquad &&\text{if}  \qquad 
\Xi(g) \prec  \Xi(f),
\\
 \lim_{t\rightarrow  - \infty}\phi_{W}(f_t)\phi_{W'}(g_t)\Omega&=\left(\phi_{W}(f )\Omega\times \phi_{W'}(g )\Omega\right)_{in} \qquad &&\text{if}  \qquad 
\Xi(f) \prec  \Xi(g) ,
\end{align*}
where we used the standard notation for collision states. Our task is now to follow similar arguments made in \cite{GL1} in order to calculate the   amplitudes of a two-particle scattering. First note that the limits will depend on the wedge as well. Moreover, our model exhibits an independence of $t\in\mathbb{R}$ for the expression $\phi_{W}(f_t)\phi_{W'}(g_t)\Omega$. This in particular lies in the definition of $\phi_{W}(f_t)$, $f^{+},\,f^{+}_{t}$ and the support properties of $\widetilde{f}$.  The particular form of the scattering states are given in the following theorem.
\newline
 
\begin{theorem}
Let $W\in\mathcal{W}_{0}$ and $f\in\mathscr{S}(\mathbb{R}^d)$ and let the unitary operator $V$  be as demanded in Proposition \ref{wl1}. 
 Then the massive deformed field  $\phi_{\theta,X}$ satisfies the properties of a tempered polarization free generator and the explicit form of two-particle scattering
states are given for test functions $f,g\in \mathscr{S}(\mathbb{R}^d)$ by 
\begin{align*}
\lim_{t\rightarrow   \infty}&\phi_{W}(f_t)\phi_{W'}(g_t)\Omega= \Gamma( V^{-1})\phi_{\theta,P}(Vf^{+})\phi_{\theta,P}(Vg^{+})\Omega\quad &&\text{if}  \qquad 
\Xi(g) \prec  \Xi(f),\\
 \lim_{t\rightarrow  - \infty}&\phi_{W}(f_t)\phi_{W'}(g_t)\Omega=\Gamma( V^{-1})\phi_{\theta,P}(Vf^{+})\phi_{\theta,P}(Vg^{+})\Omega\quad &&\text{if}  \qquad 
\Xi(f) \prec  \Xi(g).
\end{align*}

\end{theorem}
\begin{proof} 
The property of tempered polarization is  the main result of Lemma \ref{pfg}. The scattering states can be simply calculated by using the unitary equivalence given in 
Lemma \ref{ldx1e}.

\end{proof}
 
\section{Isomorphism to a Non-commutative Space-time} \label{s5}
In \cite{GL1} the deformed fields correspond to free fields defined on the representation space of the Moyal-Weyl plane  $\mathcal{V}$. This correspondence is proven by defining a unitary operator which maps the Fock space $\mathscr{H}$ to the tensor product space  $\mathcal{V}\otimes\mathscr{H}$.  Hence, the fields deformed with the momentum operator are on one hand wedge-covariant, wedge-local and nontrivial and on the other hand they correspond to free fields on a non-commutative space-time (NCST).  Therefore, one question naturally arises in the context of this more general setting. To which  NCST do fields, deformed with the unitary transformed operator, correspond to?  In this section we partially answer this question by constructing a correspondence to a NCST.
\\\\
Let $\mathcal{V}$ be  the representation space of the *-algebra which is generated by the self-adjoint operators $\hat{x}$ that fulfill the commutator relation
\begin{equation}\label{xcr1}
 [\hat{x}_{\mu},\hat{x}_{\nu}]=-2i\theta_{\mu\nu},
\end{equation}
where $\theta$ is the center of the algebra. An isomorphism exists between the $*$-algebras of fields deformed with the momentum operator and the   $*$-algebra of the free fields on non-commutative Minkowski space $\mathcal{V}$, \cite{GL1}. This equivalence is given by  the following unitary operator $V_{P,\xi}= \bigoplus_{n=0}^{\infty}V^{(n)}_{P,\xi}:\mathscr{H}\rightarrow\mathcal{V}\otimes\mathscr{H}$,  with $\xi\in\mathcal{V}$ and $||\xi||_{\mathcal{V}}=1$,
\begin{equation}\label{nc1}
\left(V^{(n)}_{P,\xi}\Psi_n\right)\left(\mathbf{p}_1,\dots,\mathbf{p}_n\right)=\Psi_n\left(\mathbf{p}_1,\dots,\mathbf{p}_n\right)\cdot e^{i\sum\limits_{k=1}^{n}p_{k}\hat{x}}\xi,\qquad \Psi_n\in\mathscr{H}_n.
\end{equation}
Hence,   the following equations hold in a distributional sense 
\begin{align}\label{aot} a_{\otimes, P}(\mathbf{p}):=e^{-ip  \hat{x} }\otimes a(\mathbf{p})= 
 V_{P,\xi} a_{\theta,P}(\mathbf{p})V^{*}_{P,\xi},
\end{align}
where an analogous relation holds for the creation operator. Moreover, it follows from $V_{\theta,\xi}\Omega=\xi\otimes\Omega$  that the $n$-point functions of $\phi_{\otimes,P}$, i.e. the free fields on non-commutative Minkowski space, coincide with the those of the deformed field $\phi_{\theta,P}$,
\begin{equation}
 \langle  (\xi\otimes\Omega),\phi_{\otimes, P}(f_1)\dots\phi_{\otimes, P}(f_n)(\xi\otimes\Omega)\rangle=  \langle \Omega, \phi_{\theta,P}(f_1)\dots
\phi_{\theta,P}(f_n)\Omega\rangle. \end{equation}
Now since we deform with operators other than the momentum operator, we should obtain an isomorphism describing the equivalence of the deformed fields with fields living on different non-commutative space-times. These space-times correspond in a certain manner to the Moyal-Weyl since we deform with operators that are unitary equivalent to the momentum operator, that  in turn generates the Moyal-Weyl spacetime. One path leading to the newly generated non-commutative space-time is by using the twist deformation (see \cite{GW1, Tu, Zah, AB, Cha, Sol} and references therein). In particular one could calculate the NC space-time by using the twisted commutator between the coordinates as   already done for the special conformal operator in \cite{MUc}. Next, we examine the equivalence of our deformed fields with the twist deformation approach.   In this context the next lemma gives a unitary operator mapping the deformed fields $\phi_{\theta,X}$ to fields on a non-commutative space. 
\begin{proposition}Let the unitary operator $\widetilde{V}_{X,\xi}= \bigoplus_{n=0}^{\infty}\widetilde{V}^{(n)}_{X,\xi}:\mathscr{H}\rightarrow\mathcal{V}\otimes\mathscr{H}$,  with $\xi\in\mathcal{V}$ and $||\xi||_{\mathcal{V}}=1$, be given by unitarily equivalence to ${V}_{P,\xi}$ as follows,
\begin{equation}
 \widetilde{V}_{X,\xi}=(\mathbb{1}_{\mathcal{V}}\otimes \Gamma( V^{-1})){V}_{P,\xi}\Gamma( V ).
 \end{equation}
Then $\widetilde{V}_{X,\xi}$   is an isomorphism of the $*$- algebras generated by the deformed fields $\phi_{\theta, X}(f)$ to unitary equivalent $*$- algebras of the unitary transformed fields on the Moyal-Weyl space.  
\end{proposition}
\begin{proof}
Prior to  the proof let us give the following expression,
\begin{align*}  
V_{P,\xi}  \,\phi_{\theta,P}(Vf)\, V^{*}_{P,\xi}= \phi_{\otimes, P} (Vf) ,
\end{align*}
where this relation can be easily seen by the virtue of Equation (\ref{nc1}). In the next step we calculate the adjoint action of  $\widetilde{V}_{X,\xi}$ on the   the  deformed fields $\phi_{\theta,X}(f)$.
\begin{align*}  
&(\mathbb{1}_{\mathcal{V}}\otimes V^{-1}){V}_{P,\xi}\underbrace{V\left(\phi_{\theta,X}(f)\right)V^{-1}}_{\phi_{\theta,P}(Vf)}V^{*}_{P,\xi}(\mathbb{1}_{\mathcal{V}}\otimes V)=& (\mathbb{1}_{\mathcal{V}}\otimes V^{-1})\phi _{\otimes, P} (Vf) (\mathbb{1}_{\mathcal{V}}\otimes V)
\end{align*} 
The equivalence can also be proven on the level of the $n$-point functions. Hence, we have shown that the deformed fields  $\phi_{\theta, X}(f)$ are unitarily equivalent to the transformed fields that live on the Moyal-Weyl space, i.e. $\phi _{\otimes, P} (Vf)$. Note that the former equations hold  in the sense of distributions.
\end{proof}$\,$\\
The following notation is self explanatory,
\begin{align}\label{tfx}  
 \phi_{\otimes, X} ( f):=
 (\mathbb{1}_{\mathcal{V}}\otimes V^{-1})\phi _{\otimes, P} (Vf) (\mathbb{1}_{\mathcal{V}}\otimes V),
\end{align} 
since we obtained this operator by an isomorphism from the deformed fields $\phi_{\theta,X}(f)$ to the tensor product space $\mathcal{V}\otimes\mathscr{H}$.
Next, we investigate in which sense our "twisted" fields $ \phi_{\otimes, X} ( f)$ fit into the framework of  twisted deformation.    In this deformation, the point-wise product of two  scalar fields is replaced by the so called twist product. Let us give a precise mathematical definition of the former statement. Let $\mu: \mathscr{S}(\mathbb{R}^d)\otimes\mathscr{S}(\mathbb{R}^d)\rightarrow\mathscr{S}(\mathbb{R}^d)$ denote the point-wise product of Schwartz functions. Then, the twisted product denoted by $\mu_{\theta,P}$ can be defined as $\mu_{\theta,P}= \mu \circ \mathcal{F}_{\theta,P}$ where 
\begin{align}\label{tf} 
\mathcal{F}_{\theta,P} = e^{-i \theta_{\mu\nu}P^{\mu}\otimes P^{\nu}}.
\end{align} 
In a remarkable paper \cite{Zah}, the author gave a rigorous meaning to the twist product of two scalar fields by going to momentum space. Moreover, it was shown that there exists an equivalence between the product of  twist deformed scalar fields and scalar fields introduced on the tensor product space $\mathcal{V}\otimes\mathscr{H}$ given in \cite{DFR}. In particular the formula   was given by 
\begin{align}    \label{tf1} 
\phi_{\otimes, P} ( f_1) \phi_{\otimes, P} ( f_2) = \phi^2_{\otimes, P}(\mu \circ \mathcal{F}_{\theta,P}(f_1\otimes f_2 )),
\end{align} 
where the following notation was introduced, 

\begin{align*}    
\phi^n_{\otimes, P} ( f)=\int \prod_{i=1}^n\, dk_i \, \left( e^{i(k_1+\dots+k_n)\hat{x} }\otimes\hat{f}(k_1,\dots,k_n) \prod_{i=1}^n \check{\phi}(k_i)\right).
\end{align*} 
\begin{remark}
The notation introduced in the context of twist-deformation is written off-shell. The reason therein lies in the extension of the twisted-QFT to scattering. However, in this work we shall proceed by going on-shell, i.e. $\check{\phi}(k)=\delta(k^2-m^2)\tilde{\phi}(k)$. 
\end{remark}
By using the former notations and products of the twisted fields we are able to give the following lemma. 
\begin{lemma}
The product of two twisted fields $\phi_{\otimes, X} ( f_1) \phi_{\otimes, X} ( f_2)$ is given by unitary equivalence to the product of two Moyal-Weyl twisted fields as follows,
\begin{align*}   
\phi_{\otimes, X} ( f_1) \phi_{\otimes, X} ( f_2)= (\mathbb{1}_{\mathcal{V}}\otimes\Gamma( V^{-1}))\phi^2_{\otimes, P}(\mu \circ \mathcal{F}_{\theta,P}(Vf_1\otimes Vf_2 ))
 (\mathbb{1}_{\mathcal{V}}\otimes \Gamma( V ) ).
\end{align*} 
Moreover, the product of $n$-twisted fields $\phi_{\otimes, X} ( f_1),\dots, \phi_{\otimes, X} ( f_n)$ is given by unitary equivalence as
\begin{align*}   
 \phi_{\otimes, X} ( f_1) \cdots\phi_{\otimes, X} ( f_n)=(\mathbb{1}_{\mathcal{V}}\otimes \Gamma( V^{-1}))\phi^n_{\otimes, P}(\mu \circ \mathcal{F}_{\theta,P}(Vf_1\otimes\dots\otimes V f_n ))
 (\mathbb{1}_{\mathcal{V}}\otimes \Gamma( V )  ).
\end{align*} 
\end{lemma}

\begin{proof}
The products simply follow by the virtue of Equation (\ref{tfx}) and by the fact that the operator $V$ is unitary. 
\end{proof}
One question still remains unsettled. How far do the twisted fields $\phi_{\otimes, X} ( f )$ correspond to the twisting deformation framework? In particular if the operators generating the twist (\ref{tf}) are unitary equivalent to the momentum operator, do  the fields $\phi_{\otimes, X} ( f )$ represent the correct twisted field according to the deformation chosen? Hence,  we take the unitarily transformed twist operator and calculate the non-commutative space-time. Next we calculate the product of two such twisted fields. Does this product correspond to the formula given in Equation (\ref{tf1})? These questions can be partially answered and are  investigated by looking at   simple examples.\newline

\begin{example}
Let the unitary operator $\Gamma( V )$ be given by the Lorentz-transformation $U(\Lambda)$. Then the twisted field $\phi_{\otimes, X}(f)$ is given by 
\begin{align*} 
 \phi_{\otimes, X} ( f)&=
 (\mathbb{1}_{\mathcal{V}}\otimes V^{-1}) \int  d^n\mu(\mathbf{k})   \left( e^{ik\hat{x} }\otimes {f^-}(\Lambda^{T}\mathbf{k})a(\mathbf{k})+h.c.\right)(\mathbb{1}_{\mathcal{V}}\otimes V)
\\&=\int  d^n\mu(\mathbf{k})   \left( e^{ik\hat{x} }\otimes {f^-}(\Lambda^{T}\mathbf{k})a(\Lambda^{T}\mathbf{k})+h.c.\right) 
\\&=\int  d^n\mu(\mathbf{k})   \left( e^{i k (\Lambda^T\hat{x}) }\otimes {f^-}( \mathbf{k})a( \mathbf{k})+h.c.\right)
, 
\end{align*} 
where in the last lines we used the explicit result of the adjoint action of $U(\Lambda)$ on the particle operators, the representation of the field in $\mathcal{V}\otimes \mathscr{H}$ (see Equation (\ref{aot}))  and the Lorentz-invariance of the measure. To compare this result  with the twisted field obtained by 
deforming with the unitarily transformed twist, we have to rewrite the coordinate operators. It can be easily done in this simple case, since the new coordinate operators are simply $\hat{x}':=\Lambda^{T} \hat{x}$ with the following commutation relations,
\begin{align*}  
 [\hat{x}'_{\mu},\hat{x}'_{\nu}]=-2i\left(\Lambda^{T} \theta\Lambda\right)_{\mu\nu}.
\end{align*}  
These are on the other hand the expected commutator relations of the coordinate operators that correspond to a NC space-time generated by the Lorentz-transformed twist. It  is as well clear that a translation of the deformed field will not be noticed on since the operators $U(x)U(\Lambda)^{-1}PU(\Lambda)U(x)^{-1}$ and $U(\Lambda)^{-1}PU(\Lambda)$ generate the same \textbf{constant} NC space-time.
\end{example}$\,$\\ In the next example we choose to work in the massless case, i.e. with the Lorentz-invariant measure $d^n\mu(\mathbf{p}):= 
d^n\mathbf{p} (2|\mathbf{p}| )^{-1}$. $\,$\\ 
\begin{example}In   the case of massless scalar fields we have an essentially self-adjoint operator $D$ and therefore a strongly continuous one parameter group $\Gamma( V )=e^{ibD}$.  Then the twisted field $\phi_{\otimes, X}(f)$ is given by 
\begin{align*} 
 \phi_{\otimes, X} ( f)&=  \int  d^n\mu(\mathbf{k})   \left( e^{i k (e^{-b}\hat{x}) }\otimes {f^-}( \mathbf{k})a( \mathbf{k})+h.c.\right)
.
\end{align*} As before, we compare the result with the twisted field obtained by  deforming with the unitarily transformed twist. Hence, we rewrite the coordinate operators as before, i.e.    $\hat{x}':=e^{-b} \hat{x}$ with the following commutation relations,
\begin{align*}  
 [\hat{x}'_{\mu},\hat{x}'_{\nu}]=-2ie^{-2b} \theta_{\mu\nu}.
\end{align*}  
These are on the other hand the expected commutator relations of the coordinate operators that correspond to a NC space-time generated by the scale-transformed twist. It  is again clear that a translation of the deformed field will not be noticed   since the operators $\Gamma( V^{-1})\,P\,\Gamma( V )$   generate a  {constant} NC space-time. Since the NC space-time is constant it does not need to be translated with regards to the Poincar\'e transformed fields. 
\end{example}$\,$\\
  However, not all unitary transformations are as simple as the Lorentz-transformations or dilatations and thus it  remains unclear to what extent the field $\phi_{\otimes, X}(f)$ corresponds to the respective twist deformation. This question shall be attacked more viciously in the context of algebraic QFT by the deformation of massless fields.
\section{Conclusion and Outlook} 
In this paper we established the existence of a broad class of deformations that result in wedge-locality and non-trivial two-particle scattering. Moreover, a new light was shed on the Poincar\'e transformational behavior of the deformed fields that correspond to fields living on a NC space-time. In fact it is required to incooperate the transformation into the quantized space-time.\\\\
A connection between the newly deformed fields and the scalar fields obtained by twist deformation was established as well. The connection is given by an explicit isomorphism. However, it is not clear if our fields defined on the tensor product space $\mathcal{V}\otimes\mathscr{H}$ correspond to fields obtained by a twist deformation, other than the constant cases,  i.e. $\Gamma( V )=U(\Lambda)$ or $\Gamma( V )=e^{ibD}$. To prove such an isomorphism we can extend the operator $ \widetilde{V}_{X,\xi}$  to $ (\tilde{U}\otimes \Gamma( V^{-1})){V}_{P,\xi}\Gamma( V )$ and construct an explicit operator $\tilde{U}$. This will have to be studied with specific cases in order to be able to achieve a generalization for  arbitrary $V$. Therefore, one should examine this isomorphism more thoroughly in the case of the massless field. The reason therein lies in the multiplicity of well studied unitary operators in the massless case that leave the wedge covariant. In particular, the conformal group provides a huge class of unitary operators for the investigation of the isomorphism to NC space-times. \\\\ The deformation was achieved with operators that are unitarily equivalent to the momentum operator. This may seem to be a restriction on the operators used for deformation. However, this is not very restrictive since by using the spectral theorem \cite[Chapter 8, Corollary 1.6]{TayII} every set of commuting self-adjoint operators can be represented by the unitary equivalence to the momentum operator. \\\\

  \bibliographystyle{alpha}
\bibliography{allliterature1}

 \end{document}